\documentclass[11pt,numbers]{article}
\pdfoutput=1
\usepackage[margin=1in]{geometry}
\usepackage{authblk}
\usepackage{fontenc}
\usepackage{amsmath}
\usepackage{amsfonts,amssymb,amsthm, bbm}
\usepackage{hyperref}
\usepackage[capitalize]{cleveref}
\usepackage{mathtools}
\usepackage{xcolor}
\usepackage{tikz}
\usepackage{braket}

\hypersetup{
    colorlinks,
    linkcolor={red!50!black},
    citecolor={blue!50!black},
    urlcolor={blue!80!black}
}
\usepackage{graphicx}   
\usepackage{subfigure}  
\usepackage{amsbsy}
\usepackage[bold]{hhtensor}
\usepackage{natbib}
\usepackage{times}
\usepackage{multirow}
\usepackage{mathrsfs}
\newtheorem{theorem}{Theorem}

\newtheorem{lemma}[theorem]{Lemma}

\newtheorem{definition}[theorem]{Definition}

\newcommand{\myleft}{\mathopen{}\mathclose\bgroup\left}
\newcommand{\myright}{\aftergroup\egroup\right}

\DeclareMathOperator{\Tr}{Tr}

\newcommand{\fermham}{\mathcal{H}_{\mathscr{S}_n}^{(f)}}

\newcommand{\orthox}{\Pi_{\bar{X}}^{(f)}}
\newcommand{\orthoz}{\Pi_{\bar{Z}}^{(f)}}
\newcommand{\rotorthox}{\tilde{\Pi}_{\bar{X}}^{(f)}}
\newcommand{\rotorthoz}{\tilde{\Pi}_{\bar{Z}}^{(f)}}
\newcommand{\rotop}{D^{(f)}}
\newcommand{\conjrotop}{{D^{(f)}}^{\dagger}}
\newcommand{\identity}{\mathbb{I}}

\newcommand{\rotfermham}{\mathcal{\tilde{H}}_{\mathscr{S}_n}^{(f)}}
\newcommand{\sumk}{\sum\limits_{j=1}^{k}}

\newcounter{example}[section]

\allowdisplaybreaks[4]

\begin{document}
\title{Constructing Fermionic Hamiltonians with Non-Gaussianic low-energy states}

\author{Kartik Anand\thanks{Email: kartik.anand.19031@iitgoa.ac.in}}

\affil{Indian Institute of Technology, Goa}
\maketitle
\begin{abstract}
Quantum PCP conjecture is one of the most influential open problems in quantum complexity theory, which states that approximating the ground state energy for a sparse local Hamiltonian upto a constant is QMA-complete. However, even though the problem remains unsolved, weaker versions of it—such as the NLTS \cite{freedman2013quantum,anshu2022nlts} and NLSS \cite{GL22} conjectures—have surfaced in the hope of providing evidence for QPCP. While the NLTS hamiltonians were first constructed in\cite{anshu2022nlts}, NLSS conjecture still remains unsolved. 
Weaker versions of the NLSS conjecture were addressed in \cite{coble2023local,coble2024}, demonstrating that Clifford and almost-Clifford states—a subclass of sampleable states—have a lower energy bound on Hamiltonians prepared by conjugating the NLTS Hamiltonians from \cite{anshu2022nlts}. In similar spirit, we construct a class of fermionic Hamilltonians for which energy of Gaussian states, a subclass of sampleable fermionic states, is bounded below by a constant. We adapt the technique used in \cite{coble2023local} to our context.
\end{abstract}

\section{Introduction}
The study of low-energy states of sparse Hamiltonians that describe most of the physical models is a central one in the literature of Quantum many-body physics. The low-energy sector largely determines the macroscopic behavior of a material such as its phase, order, and responsiveness to external perturbations. In quantum complexity literature, the complexity of the problem of approximating the ground-state energy of local Hamiltonians is provably QMA-complete for the promise gap of the order of inverse polynomial in the number of qubits\cite{KSV}. It's classical analog, Constraint Satisfaction Problem (CSP) can also be straightforwardly shown to be NP-complete for inverse polynomial promise gap. However, in the proof of the celebrated PCP theorem, CSP was shown to be NP-complete for a constant promise gap \cite{pcpproof}. Similarly, a quantum version of the PCP theorem was conjectured in \cite{qpcpconjecture} and has remained one of the most important problems of quantum complexity theory.\\
In hopes of progressing toward solving the qPCP conjecture, Hastings \cite{freedman2013quantum} proposed that if the qPCP conjecture were true, then there must exist a family of Hamiltonians such that all 'trivial' (constant depth) states are lower bounded by a constant. Along with that, this would prove that the current ansatz methods such as Variational Quantum Eigensolvers or Quantum Adiabatic Optimization Algorithm, that essentially generate trivial low-depth circuits for approximating ground state energies in quantum chemistry and quantum many-body physics—are not good enough in approximating ground-state energies of all possible local Hamiltonians.\\
\textbf{The NLTS conjecture} was shown to be true for Hamiltonians constructed from good Low-Density Parity Check (LDPC) quantum codes \cite{anshu2022nlts}. However, these Hamiltonians did not directly correspond to the Fermionic Hamiltonians, which is what nature uses to 'simulate' physical matter in its own beautiful way—where 'matter particles' that is fermions, follow anti-commutation rules in their creation-annhilation operators. Motivated by this, \cite{anshufermnlts} constructed the NLTS-fermionic Hamiltonians by using a direct n-qubits to 3n-Majorana qubit assimilation mapping which appeared in \cite{BGKT:manybody}. The mapping used was a locality preserving one unlike the standard Jordan-Wigner transformation\\
\textbf{No Low-energy Sampleable States (NLSS) conjecture}, first posited by \cite{gharibiannlss22}, suggests the nonexistence of sampleable states for a Hamiltonian class below an $\Omega(1)$ energy value, unless MA$\neq$QMA. The weaker versions of NLSS: NLCS\footnote{C stands for Clifford} \& NLACS\footnote{AC stands for Almost Clifford} were proved in \cite{coble2023local,coble2024} respectively. However, similar to \cite{anshu2022nlts}, these works focus on qubit Hamiltonians, which don't govern physical-matter quantum systems, such as various models involving electrons and chemical compounds.\\
In this work, we further extend the construction of Fermionic Hamiltonians in \cite{anshufermnlts} and provide a family of Fermionic Hamiltonians whose gaussian states' energy is lower bounded by a constant. We note that the constructed Fermionic Hamiltonians in \cite{anshufermnlts} are purely gaussanic in nature and hence the ground states of such hamiltonians are gaussian as well (particularly, with zero ground state-energy). It is widely known in the Many-body literature that gaussian states and certain non-gaussian states with \textit{'structure'} are efficiently simulable and hence we prove a weaker version of "Fermionic" NLSS conjecture which we call here, \textit{\textbf{F}}-\textbf{NLGS} (\textit{(\textbf{F}ermionic)} \textbf{N}o \textbf{L}ow-energy \textbf{G}aussian \textbf{S}tates)\footnote{Or it could simply be called NLGS, as fermionic systems and qubit matchgate circuits are equivalent. But to avoid any confusion of fermionic systems with qubit systems we prefer adding that 'F' at the front.}.\\
This paper is organized as follows. In Section \ref{sec:background} we provide a brief overview of the research around Quantum PCP conjecture while defining qPCP \& other NLTS variants. In Section \ref{sec:preliminaries} we setup notation that will be used throughout and introduce important concepts directly used in our work. Section \ref{sec:mainproof} contains the proof for our stated F-NLGS conjecture. Finally, Section \ref{sec:discuss} concludes the paper with some other discussions \& future work.

\section{Background on QPCP conjecture}
\label{sec:background}
We review Constraint Satisfaction Problem (CSP), Local Hamiltonian Problem (LHP), Quantum-PCP conjecture, NLTS and its variants and all of their relationships.\\
\textit{\textbf{Constraint Satisfaction Problem:} A constraint satisfaction problem is defined by triple $\langle V,\{D_i\}_{i=1}^{n}, C\rangle$ where $V = \{x_1,...,x_n\}$ is a set of variables with each $x_{i}$ taking values in a finite domain $D_i$ and $C$ is collection of constraints-each a relation $R \in \prod_{x \in S}D_x$ on some subset $S \subseteq V$. A solution is an assignment $f: V \rightarrow \bigcup_i D_i$ such that for every constraint $(S,R)\in C$, the tuple $(f(x)_{x \in S})$ lies in $R$}.\\
Simply put, CSP asks for a assignment such that every clause (constraint) in the relation is satisfied and thus the whole expression.\\
\textit{\textbf{Local Hamiltonian Problem (LHP-$\delta(n)$):} The k-local Hamiltonian of an n-qubit system, $\mathcal{H} = \frac{1}{m}\sum_{i=1}^{m}\mathcal{H}_i$ where $m = poly(n)$ and each $\mathcal{H}_i$ acts non-trivially on at most $k=\mathcal{O}(1)$ qubits and without loss of generality has $||\mathcal{H}_i||\leq1$. Decide whether the ground state energy of $\mathcal{H}$ is $\geq b$ or $\leq a$ (promised one of them is true), for promise gap $b-a\geq\delta(n)$}.\\
It is trivial to see that LHP-$\delta(n)$ is a quantum analog of CSP, as the solution to the former implies the corresponding ground state optimally satisfies all of the local Hamiltonian terms (constraints). CSP has been proved to be NP-complete for a promise gap $\delta(n) = \frac{1}{poly(n)}$, which is straightforward to show and also $\delta(n) = \Omega(1)$\cite{pcpproof}. Hence, it is natural to ask for its quantum analog, LHP-$\delta(n)$ is QMA-complete for $\delta(n) = \Omega(1)$, conjectured as \textit{QPCP conjecture} which has proven difficult to show.\\
\textit{\textbf{Quantum PCP Conjecture:} LHP-$\delta(n)$ is QMA-complete for $\delta = \Omega(1)$}.\\
As noticed by \cite{umeshpcp}, there is an alternative statement for QPCP conjecture, which is generally referred to as proof-checking formulation of the QPCP conjecture. It states that, one can solve any promise problem is QMA, using a QMA verifier, which only accesses a constant number of qubits from a quantum proof. However, this formulation has received considerably less attention and we point the reader to \cite{qpcpadapt} for some interesting results in that direction.\\
In hope to progress towards proving/disproving the \textit{'hardness of approximation'} version of quantum PCP, many results have surfaced providing evidence both for \cite{umeshpcp, freedman2013quantum, eh17, nvy18, anshu2022nlts, anschuetz2023combinatorial, anshufermnlts} and against \cite{bravyilocal,BH:dense-approx,aharanov} the conjecture.\\
NLTS conjecture was formulated in \cite{freedman2013quantum}, as one of the properties which a Hamiltonian must follow to be QMA-hard with a constant promise gap (under the widely believed assumption NP$\neq$QMA).\\
\begin{figure}[htbp]
  \centering
  \includegraphics[width=0.8\textwidth]{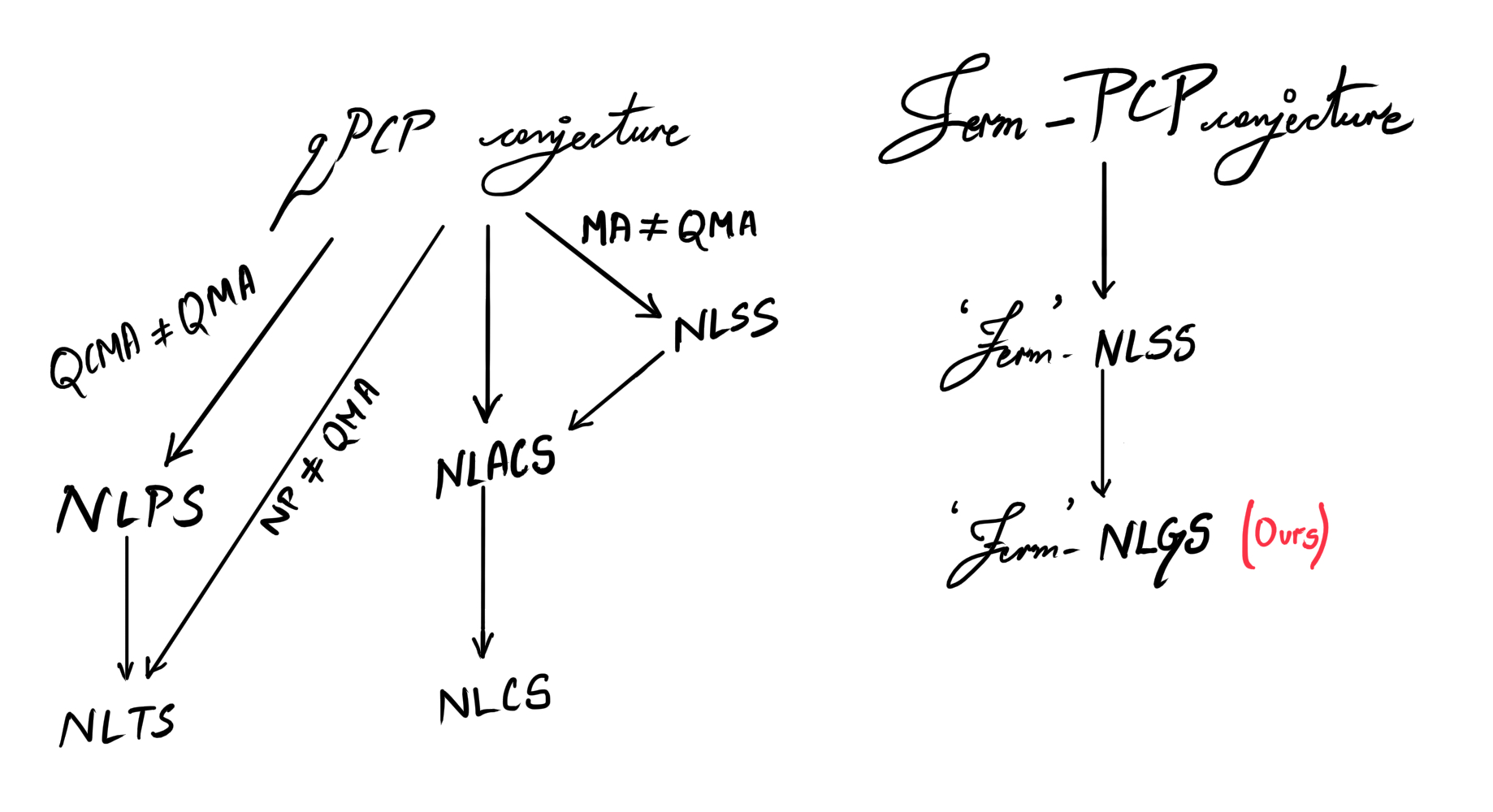}
  \caption{\small
  Relationships between qPCP \& Ferm-PCP conjectures with their respective weaker versions. 'P' in NLPS stands for 'Polynomial', 'C'/'AC' in NLCS/NLACS stands for 'Clifford'/'Almost-Clifford', 'S' in NLSS stands for 'Sampleable' and 'G' in NLGS stands for 'Gaussian'.}
  \label{fig:NLXS}
\end{figure}
\textbf{No Low-energy Trivial States (NLTS) conjecture:} A local Hamiltonian $\mathcal{H}\geq0$ has the NLTS-$\epsilon$ property, if any family of n-qubit states $\psi_n$ with energy Tr$(\psi_n\mathcal{H})\leq\epsilon n$ requires a quantum circuit, which uses an arbitrary number of ancilla qubits, of depth at least $T=\Omega(\log n)$.\\
NLTS conjecture was shown to hold true in the case of good quantum LDPC codes \cite{anshu2022nlts}. However, while the low (enough) energy states of such Hamiltonians are not constant depth, their ground state can be efficiently prepared by a Clifford circuit, thus opposing the idea of being 'complex', as Clifford circuits can be efficiently simulated with great precision. Thus, prompting \cite{GL22} to posit NLSS conjecture as the 'next goal' to prove for the quantum complexity community, under the assumption MA$\neq$QMA.\\
\textit{\textbf{No Low-energy Sampleable States (NLSS) conjecture:} There exist a family of $O(1)$-local Hamiltonians $\{\mathcal{H}_n\}_{n \in \mathbb{N}}$, where each $\mathcal{H}_n$ acts on $n$ qubits, and a constant $\epsilon > 0$ such that for any family of quantum states $\{\psi_n\}_{n \in \mathbb{N}}$ that has a succinct representation allowing sampling-access, where each $\psi_n$ is an $n$-qubit state, we have for any sufficiently large $n$:} 
$$
\Tr\left[\psi_n\mathcal{H}_n\right]> \lambda_{\mathcal{H}_n} + \epsilon.
$$
A significant progress was made towards this conjecture in \cite{coble2023local,coble2024} where the authors provably showed that by inducing 'magic' to the NLTS-Hamiltonians, it was possible to construct Hamiltonians whose low (enough) energy states are not prepared from Clifford circuits (NLCS) or Almost-Clifford circuits (NLACS), under the assumption that NP$\neq$QMA. Both these statements are pre-requisites for NLSS conjecture thus situation looks promising for the NLSS conjecture at the moment. However, these qubit Hamiltonians don't describe the Fermionic systems, which are found in many physical systems made of matter including the newest chip Majorana 1 announced by Mircrosoft, hence, we adapt this construction to this regime for constructing a family Fermionic Hamiltonians for which all gaussian states' energy is lower-bounded by a constant term. Thus, proving a 'weaker' version of NLSS conjecture in Fermionic systems as NLSS implies the existence of F-NLGS, since Gaussian states are a subset of sampleable states. We formally define the F-NLGS conjecture below.\\
\textit{\textbf{(Fermionic) No Low-energy Gaussian States (F-NLGS) Conjecture:}  
\label{def:F-NLGS}
There exists a family of Fermionic $O(1)$-local Hamiltonians $\{\mathcal{H}_n\}_{n \in \mathbb{N}}$, where each $\mathcal{H}_n$ acts on $n$ Fermionic modes, and a constant $\epsilon \geq 0$ such that for any family of quantum states $\{\psi_n\}_{n \in \mathbb{N}}$ that are mixed Gaussian states, we have for any sufficiently large $n$:  }
\[
\operatorname{Tr}(\psi_n \mathcal{H}_n) > \lambda_{\mathcal{H}_n} + \epsilon.
\]
Here, $\lambda_{\mathcal{H}_n}$ denotes the ground state energy of $\mathcal{H}_n$.

\section{Preliminaries \& Important Concepts}
\label{sec:preliminaries}
\subsection{General Definitions \& Concepts}
We assume basic familiarity with quantum computation, complexity theory and a rough idea around Fermionic systems from Quantum many-body physics literature. We denote a set of natural numbers $\{1,2,...,n\}$ with $[n]$. For subset $R\subseteq[n]$, the partial trace over qubits/modes in $R$ is denoted by $\Tr_R$, and complement of $R$ by, $-R\equiv[n] \backslash R$. Reduced state on qubits/modes on $R$ is denoted by trace over the complement of $R$ set on the full density matrix, $\Tr_{-R}[\ket{\psi}\bra{\psi}]$.\\
\begin{definition}
    \textbf{Depth-d(n) quantum circuits family:} Given a family of quantum circuits $\mathbf{C}=\{C_n\}$ such that each circuit consists of one-qubit(mode) or two-qubit(mode) gates. The number of such gates is defined as depth, $d(n)$ of the circuit. If every circuit of the family $\mathbf{C}$ has depth upper bounded by $d(n)$, $\mathbf{C}$ is called a \textit{depth-d(n) quantum circuits family}.
\end{definition}
Further, if $d(n)=O(1)$ then we say $\mathbf{C}$ is a \textit{constant-depth family of quantum circuits}. Similarly, if $d(n)=poly(n)$ then we say $\mathbf{C}$ is a \textit{polynomial-sized family of quantum circuits}.\\
\begin{definition}
    \textbf{n-qubit Pauli group}: A single qubit Pauli group is defined set $\mathcal{P}_1\equiv\{i^lP|P \in \{\identity,X,Y,Z\},l\in\{0,1,2,3\}\}$. The n-qubit Pauli group, denoted as $\mathcal{P}_n$, is the group generated by n-fold tensor product of $\mathcal{P}_1$, $\mathcal{P}_n=\bigotimes\limits_{j\in[n]}\mathcal{P}_1$.
\end{definition}
For an element $S = P_1\otimes...\otimes P_n \in \mathcal{P}_n$, the weight of S is defined to be the number of terms in the tensor product which are non-trivial upto a phase, that is, $wt(S) = |\{P_j|P_j\neq i^l\identity, l\in\{0,1,2,3\}\}|$. Such set of qubits (also modes) is denoted by $R(S)\in[n]$. We will generally denote the number of elements of this set by $k$.
\begin{definition}
    \textbf{Stabilizer group:} A list of commuting and mutually independent generators, $\mathscr{S}_n=\{S_1,...,S_k\}$, generates Stabilizer group $\langle\mathscr{S}_n\rangle$.
\end{definition}
Note that $\mathscr{S}_n$ is an abelian subgroup of $\mathcal{P}_n$, excluding $-\identity$.
\begin{definition}
\textbf{CSS code:} Given a stabilizer group $\langle\mathscr{S}_n\rangle$, corresponding Stabilizer code $\mathcal{C}_{\mathscr{S}_n}$ is defined as the common +1 eigenspace of the operators in $\langle\mathscr{S}_n\rangle$. If the generating set $\mathscr{S}_n$ contains tensor products of only Pauli $X$ and $\identity$ or only Pauli $Z$ and $\identity$, then we say $\mathcal{C}_{\mathscr{S}_n}$ is a CSS code.
\end{definition}
A code $\mathcal{C}_{\mathscr{S}_n}$ is $k'$-local if $\forall S\in\langle\mathscr{S}_n\rangle$, $S$ is non-trivial on at most $k'$ qubits and each qubit of the code is non-trivial on at most $k'$ checks $\mathscr{S}_n$. Correspondingly, the associated \textbf{$k'$-local low-density parity check Hamiltonian} is defined as
\begin{equation}
    \mathcal{H}^{(q)} = \frac{1}{|\mathscr{S}_n|}\sum\limits_{S\in\langle\mathscr{S}_n\rangle}\left(\frac{\identity - S}{2}\right)
\end{equation}
Similarly, $\mathcal{H}^{(q)}$ is a CSS (qubit) Hamiltonian if $\mathscr{S}_n$ generates a CSS code.\\
For convenience, we will denote the individual local term of CSS Hamiltonian as $\Pi_S\equiv\frac{\identity - S}{2}$, where $\Pi_S$ denotes the \textit{orthogonal} projector to $+1$ eigenspace of $S$. $\Pi_S$ acts non-trivially on $wt(S)$ qubits, hence we can also write $\Pi_S = \Pi_S|_{R(S)}\otimes\identity_{[n] \backslash R(S)}$, where $R(S)$ is the set of qubits over which $S$ is non-trivial.
We consider k-local low-density parity check (qubit/fermionic) Hamiltonians throughout this work unless stated otherwise. Such a Hamiltonian is written as $\mathcal{H}^{(q)}=\frac{1}{m}\sum\limits_{j=1}^{m}\mathcal{H}_j$ where $m=|\mathscr{S}_n|$.
For the Hamiltonians we consider in our work, we assume each local term is bounded above in spectral norm by 1, that is, $||\mathcal{H}_i||\leq1$ without loss of generality.\\
\begin{definition}
    \textbf{Frustration-free Hamiltonians:} A Hamiltonian $\mathcal{H}$ is said to frustration-free if the ground energy of $\mathcal{H}$ is $E_0\equiv\min_{\rho}\Tr[\rho\mathcal{H}]=0$ where minimization is taken over all $n$-qubit/modes mixed states.
\end{definition}
\begin{definition}
    \textbf{$\epsilon$-low-energy state:} A state $\psi$ is an $\epsilon$-low-energy state of $\mathcal{H}$ if $\Tr[\psi\mathcal{H}]<E_0+\epsilon$. For frustration-free Hamiltonians, $\Tr[\psi\mathcal{H}]<\epsilon$.
\end{definition}

\subsection{Fermionic Systems}
A system of, say, $n$ fermionic modes can be described by annhilation $a^{\dagger}_i$ and creation $a_i$ operators, for $i\in[n]$. The physics of such systems is beautifully captured by the following simple rules,
\begin{align}
    \{a_i,a_j\}=\{a^{\dagger}_i,a^{\dagger}_j\}=0 ; \{a_i,a^{\dagger}_j\}=\delta_{ij}\identity
\end{align}
However, here we concern ourselves with the Majorana description of Fermionic systems. The same system of n fermionic modes can be described by $2n$ Majorana fermion operators \textbf{(Hermitian)} $c_i$ ($i\in[2n]$) which can be obtained from annhilation-creation desciption in the following way,
\begin{align}
    a_j = \frac{1}{2}\left(c_{2j-1}+ic_{2j}\right); a^{\dagger}_j = \frac{1}{2}\left(c_{2j-1}-ic_{2j}\right)
\end{align}
and obey,
\begin{equation}
    \{c_i,c_j\}=2\delta_{ij}\identity \implies c_i^2=\identity
\end{equation}
\begin{definition}
    \textbf{(Fermionic) Gaussian states:} A state that can be realised by a Fermionic Hamiltonian that is quadratic/bilinear in Majorana operators (or annhilation-creation operators), that is, $\mathcal{H}^{(f)}=\frac{i}{2}\sum_{i,j=1}^{2k}A_{ij}c_ic_j$ is called a Gaussian state.
\end{definition}
\begin{definition}
    \textbf{Mathgate circuits:} Matchgate circuits (MGCs) consist of nearest neightbor $2$-qubit gates of the form:
    \begin{equation}
        G(A,B) = \begin{pmatrix}
            p & 0 & 0 & q\\
            0 & w & x & 0\\
            0 & y & z & 0\\
            r & 0 & 0 & s
        \end{pmatrix}; A = \begin{pmatrix}
            p & q\\ r & s
        \end{pmatrix}, B = \begin{pmatrix}
            w & x \\
            y & z \\
        \end{pmatrix}
    \end{equation}
    where $A,B\in U(2)$ and $\det(A)=\det(B)$. $A$ acts on even Hamming weight subspace of the 2-qubit inputs and $B$ on odd Hamming weight subspace.
\end{definition}
Matchgate Circuits preserve the parity of the Hamming weight of input states leading to a partition of the set of outputs of MGCs into even \& odd divisions.\\
It is widely known that outputs of MGCs acting on a fixed-computational basis input state with even \textit{(respectively odd)} Hamming weight as even \textit{(odd)} Gaussian states \cite{Terhal_2002}. These are precisely the states that arise in unassisted fermionic linear optics. Specifically, MGCs correspond to a model of non-interacting fermions in 1D. We refer the reader to \cite{Terhal_2002,bravyi00ferm,bravyi2004lagrangian,hastings21syk,HSHT} for relevant literature on this topic.\\
A surprising result concerning computational power of matchgate circuits in \cite{surprisingmatchgates} shows that the computational power of matchgate circuits is equivalent as that of universal quantum circuits which are logarithmically space bounded. However, no such equivalence is known for Clifford circuits, circuits made of quantum gates from the set $\{H, P, CNOT\}$ (where P is a single-qubit phase gate), thus as of the writing of this paper, it is speculative to believe that our result is in some way "stronger" than NLCS\cite{coble2023local}. However, both matchgate circuits and clifford circuits are efficiently classically simulable, including a few 'more complex' versions of these circuits \cite{mari2012positive, bravyi2016simulation, bravyi2016improved}.
\subsubsection*{Fermionic Mappings}
\begin{definition}
    \textbf{Jordan–Wigner Transformation (Majorana Form):}
    Consider a one-dimensional chain of $n$ qubits, with Pauli operators 
    $X_j$, $Y_j$, and $Z_j$ acting on the $j$th qubit. Define the ladder operators
    \[
    \sigma_j^{\pm} = \frac{1}{2}\Bigl( X_j \pm i\,Y_j \Bigr).
    \]
    The fermionic annihilation and creation operators, which map the $n$-qubit system to $n$ fermionic modes, are defined by
    \[
    a_j = \Bigl( \prod_{k=1}^{j-1} Z_k \Bigr) \sigma_j^{-}, \qquad
    a_j^\dagger = \Bigl( \prod_{k=1}^{j-1} Z_k \Bigr) \sigma_j^{+}.
    \]
    Next, the \textbf{Majorana operators} (here denoted by $c$) are defined as
    \[
    c_{2j-1} = a_j + a_j^\dagger 
    = \Bigl( \prod_{k=1}^{j-1} Z_k \Bigr) X_j, \qquad
    c_{2j} = i\Bigl( a_j^\dagger - a_j \Bigr) 
    = \Bigl( \prod_{k=1}^{j-1} Z_k \Bigr) Y_j.
    \]
    Conversely, the qubit (Pauli) operators can be recovered from the Majorana operators as
    \[
    X_j = \Bigl( \prod_{k=1}^{j-1} Z_k \Bigr) c_{2j-1}, \qquad
    Y_j = \Bigl( \prod_{k=1}^{j-1} Z_k \Bigr) c_{2j}, \qquad
    Z_j = -i\, c_{2j-1}\, c_{2j}.
    \]
    This transformation establishes an isomorphism between the Hilbert space of $n$ qubits and that of $n$ fermionic modes.
\end{definition}
\begin{definition}
\label{eq:qubit_assimilation}
    \textbf{Locality-Preserving Qubit Assimilation ($n$-qubits $\rightarrow$ $3n$ Majoranas)\ref{fig:qubit_assimilation}:} Let \(\mathscr{H}^{(qf)}\) be the Hilbert space of \(n\) qubits, with the standard Pauli operators
    \(X_j\), \(Y_j\), \(Z_j\) for \(j \in [n]\). 
    Then, for even \(n\), there exists a unitary map \(U\) from \(\mathscr{H}^{(qf)}\) to \(\mathscr{H}^{(f)}\) that acts \emph{locally} as follows:
    \begin{align}
        X_j &\mapsto i\, c_{2j-1}\, c_{2n+j}, \label{eq:mapX}\\[1mm]
        Y_j &\mapsto i\, c_{2j}\, c_{2n+j}, \label{eq:mapY}\\[1mm]
        Z_j &\mapsto i\, c_{2j}\, c_{2j-1}\,. \label{eq:mapZ}
    \end{align}
    Each qubit operator is mapped to a product of two Majorana operators in such a way that the mapping preserves locality.
\end{definition}
We refer the reader to \textit{Lemma 17} of \cite{anshufermnlts} for a proof of the existence of such a unitary. Another possible mapping that can be used for our construction is \cite{kitaev06} due to its locality preserving characteristics.\\
\subsubsection*{Useful Properties for Pure Gaussian states}

\begin{figure}[htbp]
  \centering
  \includegraphics[width=1.0\textwidth]{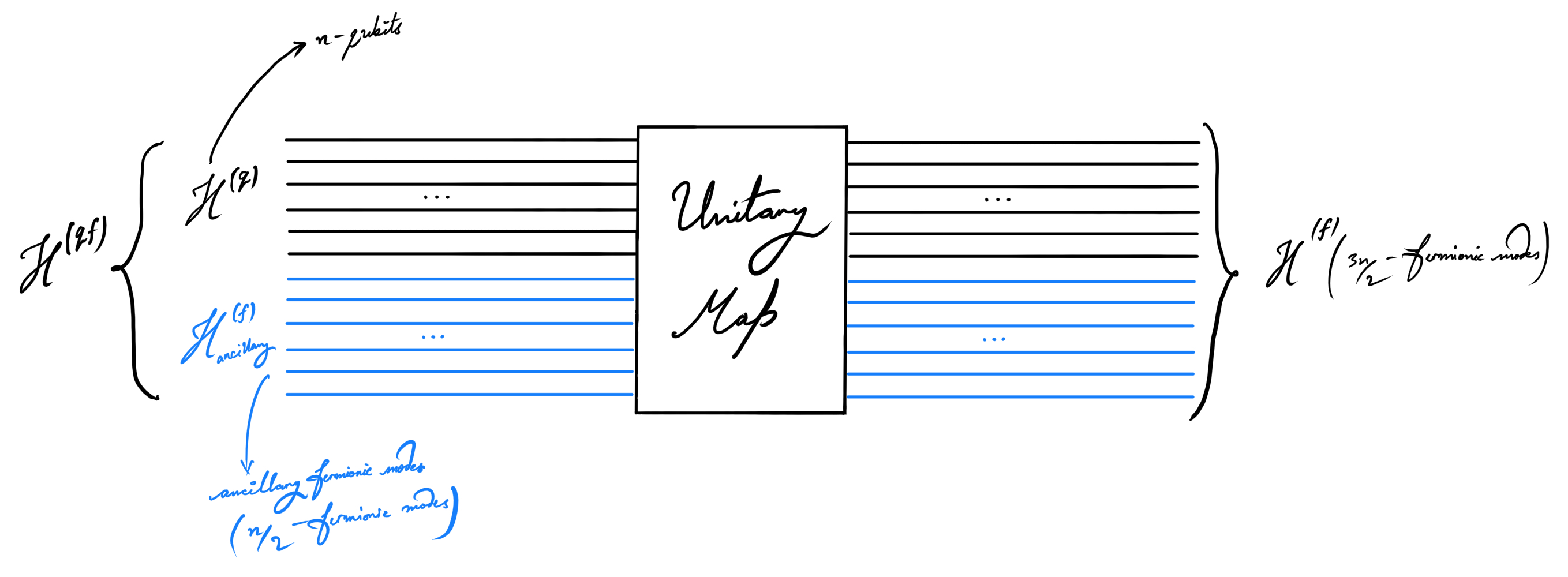}
  \caption{\small
  Illustration of Qubit Assimilation mapping used in this work. Please refer \cite{anshufermnlts} for existence of such a unitary map.}
  \label{fig:qubit_assimilation}
\end{figure}

\begin{definition}
\label{def:covariance_matrix}
    \textbf{Covariance Matrix for Majorana Fermions:} For any state $\psi$ of the system, the \emph{covariance matrix} $\Gamma$ is the $2n\times2n$ real antisymmetric matrix with entries
    $
    \Gamma_{ij} = \frac{i}{2}\, \mathrm{Tr}\Bigl(\psi\,[c_i, c_j]\Bigr) = \frac{i}{2}\, \langle [c_i, c_j] \rangle_\psi\,.
    $
    A Gaussian state is completely characterized by its Covariance matrix.
\end{definition}
The two-point Function $\langle c_ic_j\rangle_{\psi}$ or pure gaussian state $\ket{\psi}$, $\bra{\psi}c_ic_j\ket{\psi}$ is given by,
\begin{equation}
    \langle c_ic_j\rangle_{\psi}=\delta_{ij} - i\Gamma_{ij}
\end{equation}

also, for pure Gaussian states,
\begin{equation}
    \Gamma^2 = -\identity
\end{equation}

\begin{definition}
\label{def:wicks_theorem}
    \textbf{Wick's Theorem:} Let $\psi$ be a fermionic Gaussian state. Then for any four fermionic operators 
    $c_i$, $c_j$, $c_k$, and $c_l$, the expectation value satisfies
    $$
    \langle c_i\, c_j\, c_k\, c_l \rangle_\psi 
    = \langle c_i\, c_j \rangle_\psi \, \langle c_k\, c_l \rangle_\psi 
    - \langle c_i\, c_k \rangle_\psi \, \langle c_j\, c_l \rangle_\psi 
    + \langle c_i\, c_l \rangle_\psi \, \langle c_j\, c_k \rangle_\psi.
    $$
    The quartic term can decomposed into a sum of products of two-point (bilinear) correlation functions for Fermionic Gaussian states.
\end{definition}

\section{(Fermionic) No Low-energy Gaussian States}
We briefly review the construction process of Fermionic NLTS Hamiltonians in \cite{anshufermnlts}.
We start with the n-qubit NLTS Hamiltonian $\mathcal{H}_{NLTS}^{(q)} = \frac{1}{|\mathscr{S}_n|}\sum_{S\in\langle\mathscr{S}_n\rangle}\Pi_S$, where as proved in \cite{anshu2022nlts} $\mathcal{H}_{NLTS}^{(q)}$ is a CSS Hamiltonian. That is, $\Pi_S$ is an orthogonal projection of operator $S$. $S$ is either $X^{\otimes k}$ or $Z^{\otimes k}$ upto a certain permutation. We define $\bar{X}\equiv X^{\otimes k}$, $\bar{Z}\equiv Z^{\otimes k}$.\\
Based on that, we define $\Pi_{\bar{X}}$ and $\Pi_{\bar{Z}}$:
\begin{equation}
\Pi_{\bar{X}} = \frac{I - X^{\otimes k}}{2}, \quad \Pi_{\bar{Z}} = \frac{I - Z^{\otimes k}}{2}
\end{equation}

In order to transform this Hamiltonian into a Fermionic one in such a way that the Hamiltonian respects locality, unlike the widely used Jordan-Wigner transformation, and the NLTS property remains intact. As shown in \cite{anshufermnlts}, using the qubit assimilation technique \ref{eq:qubit_assimilation} transfers the NLTS property to newly constructed Fermionic Hamiltonians.\\
Using the mapping $X \rightarrow ic_{2j-1}c_{2n+j}$, $Y \rightarrow ic_{2j}c_{2n+j}$, $Z \rightarrow ic_{2j}c_{2j-1}$
\begin{equation}
\label{eqn:fermorthoprojecters}
\Pi_{\bar{X}}^{(f)} = \frac{I - \bigotimes_{j=1}^{k}(ic_{2j-1}c_{2n+j})}{2}, \quad \Pi_{\bar{Z}}^{(f)} = \frac{I - \bigotimes_{j=1}^{k}(ic_{2j}c_{2j-1})}{2}
\end{equation}
So, the NLTS Fermionic Hamiltonian looks like this:
\begin{equation}
\label{eqn:fermham}
    \mathcal{H}_{\mathscr{S}_n}^{(f)} = \frac{1}{|\mathscr{S}_n|}\sum\limits_{S\in\langle\mathscr{S}_n\rangle}\Pi_S^{(f)}
\end{equation}
Note that this Hamiltonian is for 3n/2-Fermionic mode system, 3n-Majorana system.\\
\label{sec:mainproof}
\begin{theorem}
\label{thm:main}
Given a family of $O(1)$-local Fermionic Hamiltonians $\{\fermham\}$ which are NLTS, it is possible to prepare an another family of $O(1)$-local Fermionic Hamiltonians $\{\rotfermham\}$ by conjugating the given family with a non-gaussian 'rotation operator', which are simultaneously NLTS \& NLGS.
\end{theorem}

\begin{proof}
We start with the Fermionic Hamiltonian \ref{eqn:fermham} constructed in \cite{anshufermnlts}, however, note that it's a Gaussian Hamiltonian. That is, all of the local terms are Bilinear in Majorana operators and hence, every eigenstate of this Hamiltonian is a Gaussian state, which can be efficiently simulated on a classical simulator, thus nullifying the idea of really having low energy states that are 'complex' in nature.\\
Motivated by this we adapt the technique used in \cite{coble2023local} to prepare NLCS \& NLACS qubit Hamiltonians, to construct Fermionic Hamiltonians whose energy $\forall$ Gaussian states is lower bounded by a constant.\\
The proof is mainly divided into \textbf{three} parts: \\
In Section \ref{sec:conjugation} we introduce the conjugation operation to $\fermham$ and derive the 'rotated' Hamiltonian $\rotfermham$. In Section \ref{sec:simulateouslynlts}, we present a simple proof that the constructed Hamiltonian $\rotfermham$ is also an NLTS Hamiltonian. Further, in Section \ref{sec:lower_bound} we bound the energy of low-energy Gaussian states for these Hamiltonians. Since qPCP conjecture implies the existence of simultaneous NLTS \& NLSS qubit Hamiltonians, in the same spirit, Ferm-qPCP conjecture as stated in \cite{anshufermnlts} analogously implies the existence of simultaneous NLTS \& NLGS Hamiltonians.

\subsection{Conjugating with Non-Gaussian Unitary operator}
\label{sec:conjugation}
We introduce a carefully crafted non-Gaussian rotation operator to incorporate non-Gaussianity in \ref{eqn:fermham}. For further details on the Majorana properties used in the following calculations, we refer the reader to the Preliminaries section \ref{sec:preliminaries}.\\
We denote a Non-Gaussian 'one-gate rotation operator' as $\rotop$, and the rotated Fermionic Hamiltonian $\rotfermham$ is given as:
\begin{equation*}
     \rotfermham = \rotop \fermham {\rotop}^{\dagger}
\end{equation*}
Now from \ref{eqn:fermham}, we can directly work in terms of rotating individual terms of the local Fermionic Hamiltonian $\fermham$. Thus it suffices to rotate a general $\orthox$ \& $\orthoz$ local term to talk about $\rotfermham$\\
\begin{equation*}
    \rotorthox = \rotop \orthox {\rotop}^{\dagger}; \rotorthoz = \rotop \orthoz {\rotop}^{\dagger}
\end{equation*}
We derive the expression for $\rotorthox$ and state the expression for $\rotorthoz$ as the calculations for both go hand in hand.\\
With trial and error with reckon one possible rotation operator $\rotop = exp(i\theta\sum\limits_{j=1}^{k}c_{2j-1}c_{2j+1}c_{k3}c_{k4})$ where $c_{k3}$ and $c_{k4}$ are such that, without loss of generality $\{c_{k3},c_{k4}\}=\{c_{k_l}, c_{2j}\}=\{c_{k_l}, c_{2j-1}\}=\{c_{k_l}, c_{2j+1}\}=0, \forall l\in\{3,4\}$. That is, we can always pick out indices $k_3, k_4$ such that they don't collide with the majorana operators' indices that are used in the expression of $\fermham$. To get a clear understanding of this, we refer to the diagram in the Preliminaries section \ref{sec:preliminaries}.\\
Applying this on $\orthox$ we get,
\begin{align*}
    \rotorthox = \rotop \orthox {\rotop}^{\dagger},
\end{align*}
Denote $c_{2j-1}c_{2j+1}c_{k3}c_{k4}$ by $\mathscr{C}_{K,j}$\\
We have,
\begin{align*}
    \rotorthox = \left(e^{i\theta\sum\limits_{j=1}^{k}\mathscr{C}_{K,j}}\right)\left(\frac{I - \bigotimes_{j=1}^{k}(ic_{2j-1}c_{2n+j})}{2}\right){\left(e^{i\theta\sum\limits_{j=1}^{k}\mathscr{C}_{K,j}}\right)}^{\dagger},\\
    =\frac{1}{2}\left[\mathbb{I} - \bigotimes\limits_{j=1}^{k} (i)^{k} \left(e^{i\theta\sum\limits_{j=1}^{k}\mathscr{C}_{K,j}}\right) c_{2j-1}c_{2n+j}{\left(e^{i\theta\sum\limits_{j=1}^{k}\mathscr{C}_{K,j}}\right)}^{\dagger}\right]
\end{align*}
Analyzing the Gaussian term here,
\begin{gather*}
    \left(e^{i\theta\sum\limits_{j=1}^{k}\mathscr{C}_{K,j}}\right) c_{2j-1}c_{2n+j}{\left(e^{i\theta\sum\limits_{j=1}^{k}\mathscr{C}_{K,j}}\right)}^{\dagger}=\\
    \left(e^{i\theta\mathscr{C}_{K,1}}e^{i\theta\mathscr{C}_{K,2}}...e^{i\theta\mathscr{C}_{K,j}}...e^{i\theta\mathscr{C}_{K,k}}\right).\left(c_{2j-1}c_{2n+j}\right).\left(e^{i\theta\mathscr{C}_{K,1}}e^{i\theta\mathscr{C}_{K,2}}...e^{i\theta\mathscr{C}_{K,j}}... e^{i\theta\mathscr{C}_{K,k}}\right)^{\dagger},\\
\end{gather*}
It is straightforward to show that $\mathscr{C}_{K,j}$ is a Hermitian operator and $\left[\mathscr{C}_{K,j},\mathscr{C}_{K,j'}\right]$ for $j \neq j'$. Further, we conclude that given these properties of $\mathscr{C}_{K,j}$ operator, the previous expression reduces to:
\begin{gather*}
    \left(e^{i\theta\mathscr{C}_{K,j}}\right) \left(c_{2j-1}c_{2n+j}\right) {\left(e^{i\theta\mathscr{C}_{K,j}}\right)}^{\dagger} = \left(e^{i\theta\mathscr{C}_{K,j}}\right) \left(c_{2j-1}c_{2n+j}\right) \left(e^{-i\theta\mathscr{C}_{K,j}}\right), (\because \mathscr{C}_{K,j}={\mathscr{C}_{K,j}}^{\dagger})
\end{gather*}
Define $\mathbf{c}_j \equiv c_{2j-1}c_{2n+j}$ for convenience. Before moving on, we explore the relation between $\mathbf{c}_j$ \& $\mathscr{C}_{K,j}$. It's easy to see that, $\mathscr{C}_{K,j}.\mathbf{c}_j = (c_{2j-1}c_{2j+1}c_{k_3}c_{k_4}) (c_{2j-1}c_{2n+j}) = - \mathbf{c}_j . \mathscr{C}_{K,j} \implies \left\{ \mathscr{C}_{K,j}, \mathbf{c}_j \right\} = 0 $\\
and, $ \mathscr{C}_{K,j} . \mathscr{C}_{K,j} = \mathbb{I} $.\\
Using the above relations we can finally derive the expression for $\rotorthox$:
\begin{gather}
\label{eq:rotated_orthox_term}
    \left(e^{i\theta\mathscr{C}_{K,j}}\right) \mathbf{c}_j {\left(e^{i\theta\mathscr{C}_{K,j}}\right)}^{\dagger} = \left( \cos\theta\identity + i\sin\theta\mathscr{C}_{K,j} \right) \left( \mathbf{c}_j \right) \left( \cos\theta\identity - i\sin\theta\mathscr{C}_{K,j} \right)\notag\\
    = \cos^2\theta\mathbf{c}_j - i\sin\theta\cos\theta\mathbf{c}_j\mathscr{C}_{K,j} + i\sin\theta\cos\theta\mathscr{C}_{K,j}\mathbf{c}_j + \sin^2\theta\mathscr{C}_{K,j}\mathbf{c}\mathscr{C}_{K,j}\notag\\
    =\left(\cos2\theta\identity + i\sin2\theta\mathscr{C}_{K,j}\right)\mathbf{c}_j = e^{i2\theta\mathscr{C}_{K,j}}c_{2j-1}c_{2n+j}
\end{gather}
Hence, finally we have the expression for $\rotorthox$:
\begin{gather*}
    \rotorthox = \rotop \orthox \conjrotop = \frac{1}{2}\left[ \identity - \bigotimes\limits_{j=1}^{k} e^{i2\theta\mathscr{C}_{K,j}}ic_{2j-1}c_{2n+j} \right]\\
\end{gather*}
Or, opening up the exponential term with Euler's identity will lead us to,
\begin{gather}
    \rotorthox = \rotop \orthox \conjrotop = \frac{1}{2}\left[ \identity - \bigotimes\limits_{j=1}^{k}\left( \left(\cos2\theta\right) i c_{2j-1}c_{2n+1} + \left( -\sin2\theta \right) \left(c_{2j-1}c_{2j+1}c_{k_3}c_{k_4}\right).\left(c_{2j-1}c_{2n+j}\right) \right) \right]\notag\\
    = \frac{1}{2}\left[ \identity - \bigotimes\limits_{j=1}^{k}\left( \left(\cos2\theta\right) i c_{2j-1}c_{2n+1} + \left(\sin2\theta \right) c_{2j+1}c_{k_3}c_{k_4}c_{2n+j} \right) \right] = \frac{1}{2}\left[ \identity - \bigotimes\limits_{j=1}^{k}h_{\bar{X}}^{j} \right]
\end{gather}
where we define $h_{\bar{X}}^{j} \equiv  \left(\cos2\theta\right) i c_{2j-1}c_{2n+1} + \left(\sin2\theta \right) c_{2j+1}c_{k_3}c_{k_4}c_{2n+j}$\\
Going through an almost similar calculation as above, we also obtain the expression for $\rotorthoz$. We compile both of the expressions below, where we define $h_{\bar{Z}}^{j} \equiv  \left(\cos2\theta\right) i c_{2j-1}c_{2n+1} - \left(\sin2\theta \right) c_{2j+1}c_{k_3}c_{k_4}c_{2j}$
\begin{gather}
    \rotorthox = \frac{1}{2}\left[ \identity - \bigotimes\limits_{j=1}^{k}h_{\bar{X}}^{j} \right] ; \rotorthoz = \frac{1}{2}\left[ \identity - \bigotimes\limits_{j=1}^{k}h_{\bar{Z}}^{j} \right]
\end{gather}
The local Hamiltonian terms in $\rotfermham$ are composed of terms that look like above, with each term having their own value for number of Fermionic modes over which it acts non-trivially, $k$.

\subsection{Simultaneous F-NLGS \& NLTS}
\label{sec:simulateouslynlts}
As discussed before, analogously to qubit-qPCP conjecture, Ferm-qPCP conjecture implies the existence of simultaneous NLTS and F-NLGS Fermionic Hamiltonians. Here we provide a simple argument as to why our constructed Hamiltonians when created using conjugation with a Non-Gaussian, constant depth unitary, preserves the NLTS property. Formally, we prove the following lemma which closely follows \textit{Lemma 4.} in \cite{coble2023local}.
\begin{lemma}
\label{lemma:simulatenoslynlts}
    Given $\left\{\fermham\right\}$ a family of $\epsilon$-NLTS local Hamiltonians, where $n$ is the number of Fermionic modes \& $W = \{W_n\}$ is a family of constant-depth Non-Gaussian circuits. Then, the family of conjugated Hamiltonians $\left\{\mathcal{H}_n^{(f)}\right\}$ is also $\epsilon$-NLTS.
\end{lemma}
\begin{proof}
    For the sake of the argument, say $\{\rotfermham\}$ is not NLTS. By definition, for every $\epsilon'>0$ there is an $n$ and a constant depth, Non-Gaussian circuit such that $U_{\epsilon',n}\ket{0}^{\otimes n}$ is an $\epsilon'$-low energy state of $\mathcal{\tilde{H}}_n^{(f)} = W_n^{\dagger}\mathcal{H}_n^{(f)}W_n$, i.e., $ \bra{0}^{\otimes n} U_{\epsilon',n}^{\dagger} W_n^{\dagger}\mathcal{H}_n^{(f)}W_n U_{\epsilon',n} \ket{0}^{\otimes n} < E_0(W_n^{\dagger}\mathcal{H}_n^{(f)}W_n) + \epsilon'$\\
    Since, $W_n^{\dagger}$ is a unitary operator, ground state energy. $E_0$ of $\mathcal{H}_n^{(f)}$ \& $ \mathcal{\tilde{H}}_n^{(f)} $ are equal.\\
    $\therefore$ defining $\ket{\Psi_{\epsilon, n}}\equiv W_n U_{\epsilon,n}\ket{0}^{\otimes n}$, we have,
    \begin{equation*}
        \bra{\Psi_{\epsilon,n}}\mathcal{H}_n^{(f)}\ket{\Psi_{\epsilon,n}} < E_0\left(\mathcal{H}_n^{(f)}\right) + \epsilon
    \end{equation*}
    Thus $\ket{\Psi_{\epsilon,n}}$ is an $\epsilon$-low energy state. Since $W_nU_{\epsilon,n}$ is a constant depth circuit.\\
    Implying, $\mathcal{H}_n^{(f)}$ has an $\epsilon$-low energy trivial state, contradicting assumption of $\epsilon$-NLTS
\end{proof}

\subsection{Energy lower bound on Gaussian states}
\label{sec:lower_bound}
Finally having the expression of the rotated Hamiltonian $\rotfermham$ and proving the 'NLTS preservation' of conjugation operation, we are finally read to prove our main result \textit{Theorem \ref{thm:main}}.\\
The system under consideration has $3n/2$ - fermionic modes with Hamiltonian $\rotfermham$:
\begin{equation}
    \rotfermham = \frac{1}{|\mathscr{S}_n|}\sum\limits_{S\in\langle\mathscr{S}_n\rangle}\Pi_S^{(f)}
\end{equation}
A Gaussian state for this system, $\psi$ on $3n/2$ - Fermionic modes, $\psi = \sum_jp_j\ket{\varphi_j}\bra{\varphi_j}$, is a mixture of pure gaussian states $\ket{\varphi_j}$ on $3n/2$ modes.\\
The energy of this state is given by,
\begin{gather}
\label{eq:energyexp}
    \Tr\left[\psi\fermham\right] = \sum\limits_{j}p_j\bra{\varphi_j}\fermham\ket{\varphi_j} = \frac{1}{|\mathscr{S}_n|}\sum\limits_{S\in\mathscr{S}_n}\sum\limits_{j}p_j\bra{\varphi_j}\Pi_S^{(f)}\ket{\varphi_j}
\end{gather}
With this little argument, it suffices to show a lower bound on 'local energy' of every possible pure gaussian state on $3n/2$ - Fermionic modes, that is, every $\bra{\varphi_j}\Pi_S^{(f)}\ket{\varphi_j}$ term where $\ket{\varphi_j}$ is a pure gaussian state for $3n$-Majorana system.\\
\begin{lemma}
\label{lemma:ortholowerbound}
    There exists a lower bound $L_k$ on $\bra{\eta}\tilde{\Pi}_S^{(f)}\ket{\eta}$ where $\ket{\eta}$ is a pure gaussian state, where $k$ is the number of Fermionic modes over which $\tilde{\Pi}_S^{(f)}$ acts non-trivially.
\end{lemma}
\begin{proof}
    Say the reduced set of Fermionic modes is defined as $R_S \subset \left[n\right]$, such that, $\Pi_S^{(f)}$ only works non trivially on the modes contained in $R_S$ set.\\
    \begin{gather*}
        \Pi_S^{(f)} = \Pi_S^{(f)}|_{R_S} \otimes \identity_{\left[n\right] \backslash R_S},
    \end{gather*}
    So, $\bra{\eta}\tilde{\Pi}_S^{(f)}\ket{\eta} = \Tr\left[\eta_{R_S}\tilde{\Pi}_S^{(f)}|_{R_S} \right]$ where $\eta_{R_S} \equiv \Tr_{-R_S}\left[\ket{\eta}\bra{\eta} \right]$. Since $\ket{\eta}$, a is a 3n/2 - Fermionic pure gaussian state, its reduction on $R_S$ modes is a pure gaussian state too.\\
    From above, we have, 
    \begin{equation}
    \label{eq:kgaussequals3n2gauss}
    \bra{\eta}\tilde{\Pi}_S^{(f)}\ket{\eta} = \bra{\eta_{R_S}}\tilde{\Pi}_{S}^{(f)}\ket{\eta_{R_S}}. 
    \end{equation}
    Thus, for the remaining proof we'll find it suffice to show properties such as energy lower bounds for reduced pure gaussian states on 'non-trivial modes'.\\
    We arrive at the core part of showing our main result. For the rest part of the proof we will denote $\ket{\eta_{R_S}}$ as $\ket{\xi}$. Since $\tilde{\Pi}_S^{(f)}$ is either $\rotorthox$ or $\rotorthoz$ upto a certain permutation, we show an energy lower bound for both of these.\\
    \begin{gather}
    \label{eq:individual_energy}
        \bra{\xi}\rotorthox\ket{\xi} = \frac{1}{2}\left[\identity - \bra{\xi}\bigotimes\limits_{j=1}^{k}h_{\bar{X}}^{j}\ket{\xi} \right]; \bra{\xi}\rotorthoz\ket{\xi} = \frac{1}{2}\left[\identity - \bra{\xi}\bigotimes\limits_{j=1}^{k}h_{\bar{Z}}^{j}\ket{\xi} \right]
    \end{gather}
    To lower bound these quantities, it makes sense to upper bound the non-trivial terms inside these expressions. Let's start with $\bra{\xi}\bigotimes\limits_{j=1}^{k}h_{\bar{X}}^{j}\ket{\xi}$,
    \begin{align*}
        \bra{\xi}\bigotimes\limits_{j=1}^{k}h_{\bar{X}}^{j}\ket{\xi} = \sum\limits_{j=1}^{k}\bra{\xi}h_{\bar{X}}^{j}\ket{\xi} = \sumk \bra{\xi}\left[\left(\cos2\theta\right)ic_{2j-1}c_{2n+j}+\left(sin2\theta\right)c_{2j-1}c_{k_3}c_{k_4}c_{2n+j}\right]\ket{\xi}\\
        =\sumk\left(i \cos2\theta \right) \bra{\xi}c_{2j-1}c_{2n+j}\ket{\xi} + \left(\sin2\theta\right)\bra{\xi}c_{2j+1}c_{k_3}c_{k_4}c_{2n+j}\ket{\xi}
    \end{align*}
    The foremost inequality follows from the fact that energy term for each fermionic mode is additive in nature. The terms $\bra{\xi}c_{2j-1}c_{2n+j}\ket{\xi}$ and $\bra{\xi}c_{2j+1}c_{k_3}c_{k_4}c_{2n+j}\ket{\xi}$ for each $j \in \left[n\right]$ can be shown to be lower bounded by using the unique properties of the \textit{covariance matrix} for gaussian states and \textit{Wick's theorem}. We refer the reader to preliminaries section \ref{sec:preliminaries} for more information.\\
    We have,
    \begingroup
    \footnotesize
    \begin{gather}
        \sumk\bra{\xi}h_{\bar{X}}^j\ket{\xi} = \left(\cos2\theta\right) \sumk \Gamma_{2j-1,2n+j} + \left(\sin2\theta\right)\sumk\left(-\Gamma_{2j+1,k_3}\Gamma_{k_4,2n+j} + \Gamma_{2j+1,k_3}\Gamma_{k_4,2n+j} - \Gamma_{2j+1,2n+j}\Gamma_{k_3,2n+j} \right)\notag
    \end{gather}
    \endgroup
    Using the fact that, for all pure gaussian states absolute values of all the entries of covariance matrix is upper bounded by 1, we have a rough upper bound,
    \begin{equation}
        \label{eq:upperx}
        \bra{\xi}\bigotimes\limits_{j=1}^{k}h_{\bar{X}}^j\ket{\xi} \leq k\left(\cos2\theta + 3\sin2\theta \right) \leq k_0\left(\cos2\theta + 3\sin2\theta \right)
    \end{equation}
    where $k_0$ is the largest "non-triviality value" for a local term that appears in the Hamiltonian expression $\rotfermham$ hence $k \leq k_0$. Using a similar argument we have the same upper bound for
    \begin{equation}
        \label{eq:upperz}
        \bra{\xi}\bigotimes\limits_{j=1}^{k}h_{\bar{Z}}^j\ket{\xi} \leq k_0\left(\cos2\theta + 3\sin2\theta \right)
    \end{equation}
    With \ref{eq:upperx}, \ref{eq:upperz} and \ref{eq:kgaussequals3n2gauss} for \ref{eq:individual_energy} we have,
    \begin{equation}
    \begin{aligned}
        \bra{\eta}\rotorthox\ket{\eta} &\geq \frac{1}{2}\left[1 - k_0\left(\cos2\theta + 3\sin2\theta\right) \right] \\
        \bra{\eta}\rotorthoz\ket{\eta} &\geq \frac{1}{2}\left[1 - k_0\left(\cos2\theta + 3\sin2\theta\right) \right]
    \end{aligned}
    \end{equation}
\end{proof}
Continuing with our main proof, we use the results of \ref{lemma:ortholowerbound} in \ref{eq:energyexp} to obtain our final lower bound on the energy of all (possibly mixed) gaussian state $\psi$:
\begin{gather}
    \Tr\left[\psi\mathcal{H}_{\mathscr{S}_n}^{(f)}\right] = \frac{1}{|\mathscr{S}_n|}\sum\limits_{S\in\langle\mathscr{S}_n\rangle}\sum\limits_{j}p_j\bra{\varphi_j}\Pi_S^{(f)}\ket{\varphi_j}\geq \frac{1}{|\mathscr{S}_n|}\sum\limits_{S\in\langle\mathscr{S}_n\rangle}\sum\limits_jp_j\frac{1}{2}\left[1 - k_0\left(\cos2\theta + 3\sin2\theta\right)\right]\notag\\
    = \frac{1}{|\mathscr{S}_n|}\sum\limits_{S\in\langle\mathscr{S}_n\rangle}\frac{1}{2}\left[1 - k_0\left(\cos2\theta + 3\sin2\theta\right)\right] = \frac{1}{2}\left[1 - k_0\left(\cos2\theta + 3\sin2\theta\right)\right]
\end{gather}
As we are dealing with $O(1)$-local Fermionic Hamiltonians, $k_0$ is a constant and does not scale with number of Fermionic modes. Now, as a 'final touch' we also need to be careful about choosing the parameter $\theta$ so that our energy lower bound is non-negative and $k_0\geq2$,
\begin{gather}
    \frac{1}{2}\left[1 - k_0\left(\cos2\theta + 3\sin2\theta\right)\right] > 0,\\
    \frac{1}{\left(\cos2\theta + 3\sin2\theta\right)} > k_0 \geq 2
\end{gather}
Suggesting,
$$
\sin(2\theta + \phi) \in \left(0, \frac{1}{k_0\sqrt{10}}\right), \quad \text{with } \phi = \arcsin\left(\frac{1}{\sqrt{10}}\right)
$$
Thus, for such a value of $\theta$ we have a family of Fermionic Hamiltonians $\left\{\rotfermham\right\}$ obtained from rotating the family of Fermionic NLTS Hamiltonians $\left\{\fermham\right\}$ as obtained in \cite{anshufermnlts}.

\end{proof}

\section{Acknowledgements}
We thank Anurag Anshu \& Sevag Gharibian for making their lectures on Quantum computing \& Quantum complexity theory online.

\section{Discussion}
\label{sec:discuss}
By carefully adapting the technique used in \cite{coble2023local,coble2024} to the fermionic domain, we provide the construction of fermionic Hamiltonians whose sufficiently low-energy states are non-Gaussian in nature, thus extending the work on simple stabilizer fermionic Hamiltonians presented in \cite{anshufermnlts}. We believe that the $n$-qubit $\rightarrow$ $4n$-Majoranas mapping used in \cite{kitaev06} will also lead to similar results. Our work constitutes a step forward towards qPCP conjecture (Ferm-PCP, to be precise) and provides a weaker construction of (Fermionic) NLSS Hamiltonians. In comparison to \cite{coble2023local}, our result is somewhat stronger in the sense that a reduction of Matchgate circuits to logarithmically bounded general qubit systems is well known \cite{surprisingmatchgates} unlike for Clifford circuits, to the best of our knowledge. However, we acknowledge that both Gaussian and Clifford circuits are efficiently simulable. Furthermore, Gaussian states are only a subset of sampleable states; thus, more rigorous and novel techniques are required for a thorough treatment of the conjecture. The natural direction from here is to analyze the techniques used in this work and perhaps introduce new methods to construct Hamiltonian families whose low-energy states are more complex than Gaussian states, while simultaneously satisfying NLTS and NLGS.
\newpage

\appendix

\newcommand{\etalchar}[1]{$^{#1}$}


\begin{thebibliography}{CSBM{\etalchar{+}}23}

\bibitem[AALV08]{umeshpcp}
Dorit Aharonov, Itai Arad, Zeph Landau, and Umesh Vazirani.  
\newblock The Detectability Lemma and Quantum Gap Amplification.  
\newblock {\em arXiv preprint}, arXiv:0811.3412, 2008.  
\newblock URL: \url{https://arxiv.org/abs/0811.3412}.

\bibitem[AALV09]{AALV:detect}
Dorit Aharonov, Itai Arad, Zeph Landau, and Umesh Vazirani.
\newblock The detectability lemma and quantum gap amplification.
\newblock In {\em Proceedings of the Forty-First Annual ACM Symposium on Theory
  of Computing}, STOC '09, pages 417–426, New York, NY, USA, 2009. Association
  for Computing Machinery.
\newblock URL: \url{https://doi.org/10.1145/1536414.1536472}.

\bibitem[AAV13]{AAV}
Dorit Aharonov, Itai Arad, and Thomas Vidick.
\newblock The quantum {PCP} conjecture.
\newblock {\em CoRR}, abs/1309.7495, 2013.
\newblock URL: \url{http://arxiv.org/abs/1309.7495}.

\bibitem[AE13]{qpcpconjecture}
Dorit Aharonov and Lior Eldar.
\newblock The Quantum PCP Conjecture.
\newblock {\em arXiv preprint arXiv:1309.7495}, 2013.
\newblock URL: \url{https://arxiv.org/abs/1309.7495}.

\bibitem[AE13]{aharanov}
Dorit Aharonov and Lior Eldar.
\newblock The commuting local Hamiltonian on locally-expanding graphs is in NP.
\newblock 2013.
\newblock arXiv:1311.7378.
\newblock URL: \url{https://arxiv.org/abs/1311.7378}.

\bibitem[AAG22]{anshu22area}
Anurag Anshu, Itai Arad, and David Gosset.
\newblock An area law for 2d frustration-free spin systems.
\newblock In {\em Proceedings of the 54th Annual {ACM} {SIGACT} Symposium on
  Theory of Computing}, 2022.
\newblock URL: \url{https://doi.org/10.1145%2F3519935.3519962}.

\bibitem[ABN22]{anshu2022nlts}
A.~Anshu, N.P. Breuckmann, and C.~Nirkhe.
\newblock Nlts {H}amiltonians from good quantum codes.
\newblock In {\em 55th Annual ACM Symposium on Theory of Computing (STOC
  2023)}, pages 1090--1096. ACM, 2022.
\newblock URL: \url{https://dl.acm.org/doi/10.1145/3564246.3585114}.

\bibitem[AGK23]{anschuetz2023combinatorial}
Eric~R. Anschuetz, David Gamarnik, and Bobak Kiani.
\newblock Combinatorial {NLTS} from the overlap gap property, 2023.
\newblock URL: \url{https://doi.org/10.48550/arXiv.2304.00643}.

\bibitem[AS98]{pcpproof}
S.~Arora and S.~Safra,  
``Probabilistic checking of proofs: a new characterization of NP,''  
\textit{J. ACM}, vol. 45, no. 1, pp. 70–122, Jan. 1998,  
doi: \href{https://doi.org/10.1145/273865.273901}{10.1145/273865.273901}.


\bibitem[BCDZ99]{BuhrmanCWZ99}
H.~{Buhrman}, R.~{Cleve}, R.~{De Wolf}, and C.~{Zalka}.
\newblock Bounds for small-error and zero-error quantum algorithms.
\newblock In {\em Proc. of FOCS 99'}, pages 358--368, 1999.

\bibitem[BHW24]{qpcpadapt}
Harry Buhrman, Jonas Helsen, and Jordi Weggemans.  
\newblock Quantum PCPs: on adaptivity, multiple provers and reductions to local Hamiltonians.  
\newblock {\em arXiv preprint}, arXiv:2403.04841, 2024.  
\newblock URL: \url{https://arxiv.org/abs/2403.04841}.

\bibitem[BH13]{BH:dense-approx}
Fernando~G.S.L. Brandao and Aram~W. Harrow.
\newblock Product-state approximations to quantum ground states.
\newblock In {\em Proceedings of the Forty-Fifth Annual ACM Symposium on Theory
  of Computing}, STOC '13, pages 871–880, New York, NY, USA, 2013. Association
  for Computing Machinery.
\newblock URL: \url{https://doi.org/10.1145/2488608.2488719}.

\bibitem[Bra05]{bravyi2004lagrangian}
Sergey Bravyi.
\newblock Lagrangian representation for fermionic linear optics.
\newblock {\em Quantum Info. Comput.}, 5(3):216–238, May 2005.
\newblock URL: \url{https://arxiv.org/abs/quant-ph/0404180}.

\bibitem[Bra06]{bravyi:uni}
Sergey Bravyi.
\newblock Universal quantum computation with the $\nu=5/2$ fractional quantum
  hall state.
\newblock {\em Phys. Rev. A}, 73:042313, Apr 2006.
\newblock URL: \url{https://link.aps.org/doi/10.1103/PhysRevA.73.042313}.

\bibitem[BG17]{BG:impurity}
Sergey Bravyi and David Gosset.
\newblock Complexity of quantum impurity problems.
\newblock {\em Communications in Mathematical Physics}, 356(2):451--500, Aug 2017.
\newblock URL: \url{https://doi.org/10.1007%2Fs00220-017-2976-9}.

\bibitem[BGKT19]{BGKT:manybody}
Sergey Bravyi, David Gosset, Robert K{\"o}nig, and Kristan Temme.
\newblock Approximation algorithms for quantum many-body problems.
\newblock {\em Journal of Mathematical Physics}, 60(3):032203, 2019.
\newblock URL: \url{https://doi.org/10.1063/1.5085428}.

\bibitem[BGS16]{bravyi2016improved}
S.~Bravyi, D.~Gosset, and J.~A.~Smolin.  
\newblock Improved classical simulation of quantum circuits dominated by Clifford gates.  
\newblock {\em Physical Review Letters}, 116(25):250501, 2016.

\bibitem[BK02]{bravyi00ferm}
Sergey Bravyi and Alexei Kitaev.
\newblock Fermionic quantum computation.
\newblock {\em Annals of Physics}, 298(1):210--226, 2002.
\newblock URL: \url{https://doi.org/10.1006%2Faphy.2002.6254}.

\bibitem[BSS16]{bravyi2016simulation}
S.~Bravyi, D.~Smith, and J.~A.~Smolin.  
\newblock Simulation of quantum circuits by stabilizer decompositions.  
\newblock {\em Physical Review Letters}, 117:210402, 2016.

\bibitem[BV05]{bravyilocal}
Sergey Bravyi and Mikhail Vyalyi.
\newblock Commutative version of the k-local Hamiltonian problem and common eigenspace problem.
\newblock {\em Quantum Information and Computation}, 5(3):187--215, 2005.
\newblock URL: \url{https://arxiv.org/abs/quant-ph/0308021}.


\bibitem[CRSS97]{CRSS}
A.~R. Calderbank, E.~M. Rains, P.~W. Shor, and N.~J.~A. Sloane.
\newblock Quantum error correction and orthogonal geometry.
\newblock {\em Phys. Rev. Lett.}, 78:405--408, Jan 1997.
\newblock URL: \url{https://link.aps.org/doi/10.1103/PhysRevLett.78.405}.

\bibitem[CCNN23]{coble2023local}
Nolan~J Coble, Matthew Coudron, Jon Nelson, and Seyed~Sajjad Nezhadi.
\newblock Local hamiltonians with no low-energy stabilizer states.
\newblock {\em arXiv preprint arXiv:2302.14755}, 2023.
\newblock URL: \url{https://doi.org/10.48550/arXiv.2302.14755}.

\bibitem[CCNN24]{coble2024}
Nolan~J Coble, Matthew Coudron, Jon Nelson, and Seyed~Sajjad Nezhadi.
\newblock Hamiltonians whose low-energy states require $\Omega(n)$ T gates.
\newblock {\em arXiv preprint arXiv:2310.01347}, 2024.
\newblock URL: \url{https://doi.org/10.48550/arXiv.2310.01347}.

\bibitem[CSBM+23]{chien2023simulating}
Riley~W Chien, Kanav Setia, Xavier Bonet-Monroig, Mark Steudtner, and James~D
  Whitfield.
\newblock Simulating quantum error mitigation in fermionic encodings.
\newblock {\em arXiv preprint arXiv:2303.02270}, 2023.
\newblock URL: \url{https://doi.org/10.48550/arXiv.2303.02270}.


\bibitem[EH17]{eh17}
Lior Eldar and Aram W. Harrow.  
\newblock Local Hamiltonians whose ground states are hard to approximate.  
\newblock In {\em Proceedings of the 58th Annual IEEE Symposium on Foundations of Computer Science (FOCS)}, pages 427–438, October 2017.  
\newblock URL: \url{http://dx.doi.org/10.1109/FOCS.2017.46}.


\bibitem[FGLSS96]{pcptheorem}
Uriel Feige, Shafi Goldwasser, László Lovász, Shmuel Safra, and Mario Szegedy.
\newblock Interactive proofs and the hardness of approximating cliques.
\newblock {\em Journal of the ACM}, 43(2):268–292, March 1996.
\newblock DOI: \url{https://doi.org/10.1145/226643.226652}.

\bibitem[FH13]{freedman2013quantum}
M.~H. Freedman and M.~B. Hastings.
\newblock Quantum systems on non-$k$-hyperfinite complexes: A generalization of
  classical statistical mechanics on expander graphs, 2013.
\newblock URL: \url{https://doi.org/10.48550/arXiv.1301.1363}.


\bibitem[GG22]{GL22}
Sevag Gharibian and Francois~Le Gall.
\newblock Dequantizing the quantum singular value transformation: hardness and
  applications to quantum chemistry and the quantum {PCP} conjecture.
\newblock In {\em Proceedings of the 54th Annual {ACM} {SIGACT} Symposium on
  Theory of Computing}, ACM, Jun 2022.
\newblock URL: \url{https://doi.org/10.1145%2F3519935.3519991}.

\bibitem[GL23]{gharibiannlss22}
Sevag Gharibian and Francois Le Gall.  
\newblock Dequantizing the quantum singular value transformation: Hardness and applications to quantum chemistry and the quantum PCP conjecture.  
\newblock {\em SIAM Journal on Computing}, 52(4):1009–1038, August 2023.  
\newblock URL: \url{http://dx.doi.org/10.1137/22M1513721}.


\bibitem[HO22]{hastings21syk}
Matthew~B. Hastings and Ryan O'Donnell.
\newblock Optimizing strongly interacting fermionic {H}amiltonians.
\newblock In {\em Proceedings of the 54th Annual ACM SIGACT Symposium on Theory
  of Computing}, STOC 2022, pages 776–789, New York, NY, USA, 2022.
\newblock URL: \url{https://doi.org/10.1145/3519935.3519960}.

\bibitem[HATH23]{anshufermnlts}
Yaroslav Herasymenko, Anurag Anshu, Barbara Terhal, and Jonas Helsen.  
\newblock Fermionic Hamiltonians without trivial low-energy states.  
\newblock {\em arXiv preprint}, arXiv:2307.13730, 2023.  
\newblock URL: \url{https://arxiv.org/abs/2307.13730}.

\bibitem[HSHT22]{HSHT}
Yaroslav Herasymenko, Maarten Stroeks, Jonas Helsen, and Barbara Terhal.
\newblock Optimizing sparse fermionic hamiltonians.
\newblock {\em arXiv preprint arXiv:2211.16518}, 2022.
\newblock URL: \url{https://doi.org/10.48550/arXiv.2211.16518}.


\bibitem[JSK+18]{Jiang_2018}
Zhang Jiang, Kevin~J. Sung, Kostyantyn Kechedzhi, Vadim~N. Smelyanskiy, and
  Sergio Boixo.
\newblock Quantum algorithms to simulate many-body physics of correlated
  fermions.
\newblock {\em Physical Review Applied}, 9(4), April 2018.
\newblock URL: \url{https://doi.org/10.1103%2Fphysrevapplied.9.044036}.

\bibitem[JKMW09]{surprisingmatchgates}
Richard Jozsa, Barbara Kraus, Akimasa Miyake, and John Watrous.  
\newblock Matchgate and space-bounded quantum computations are equivalent.  
\newblock {\em Proceedings of the Royal Society A: Mathematical, Physical and Engineering Sciences}, 466(2115):809–830, November 2009.  
\newblock URL: \url{http://dx.doi.org/10.1098/rspa.2009.0433}.


\bibitem[KLS96]{KahnLS96}
Jeff Kahn, Nathan Linial, and Alex Samorodnitsky.
\newblock Inclusion-exclusion: Exact and approximate.
\newblock {\em Combinatorica}, 16(4):465--477, 1996.
\newblock \href {https://doi.org/10.1007/BF01271266}
  {\path{doi:10.1007/BF01271266}}.

\bibitem[KLM+09]{kavitha2009cycle}
Telikepalli Kavitha, Christian Liebchen, Kurt Mehlhorn, Dimitrios Michail,
  Romeo Rizzi, Torsten Ueckerdt, and Katharina~A Zweig.
\newblock Cycle bases in graphs characterization, algorithms, complexity, and
  applications.
\newblock {\em Computer Science Review}, 3(4):199--243, 2009.
\newblock URL: \url{https://doi.org/10.1016/j.cosrev.2009.08.001}.

\bibitem[KSV02]{KSV}
A.~Yu. Kitaev, A.H. Shen, and M.N. Vyalyi.
\newblock {\em Classical and Quantum Computation. Vol. 47 of Graduate Studies
  in Mathematics.}
\newblock American Mathematical Society, Providence, RI, 2002.

\bibitem[Kit02]{kitaevqma}
A. Kitaev, A. Shen, and M. Vyalyi.  
\newblock Classical and Quantum Computation.  
\newblock {\em The American Mathematical Monthly}, 110, December 2003.  
\newblock DOI: \url{10.2307/3647986}.

\bibitem[Kit06]{kitaev06}
Alexei Kitaev.
\newblock Anyons in an exactly solved model and beyond.
\newblock {\em Annals of Physics}, 321(1):2--111, 2006.
\newblock URL: \url{https://doi.org/10.1016\%2Fj.aop.2005.10.005}.


\bibitem[LZ22]{leverrier2022quantum}
Anthony Leverrier and Gilles Z{\'e}mor.
\newblock Quantum tanner codes.
\newblock In {\em 2022 IEEE 63rd Annual Symposium on Foundations of Computer
  Science (FOCS)}, pages 872--883. IEEE, 2022.
\newblock URL: \url{https://doi.org/10.1109/FOCS54457.2022.00117}.

\bibitem[LCV07]{Liu_2007}
Yi-Kai Liu, Matthias Christandl, and Frank Verstraete.
\newblock Quantum computational complexity of the n-representability problem:
  Qma complete.
\newblock {\em Physical review letters}, 98(11):110503, 2007.
\newblock URL: \url{https://doi.org/10.1103/PhysRevLett.98.110503}.


\bibitem[Mar82]{margulis1982explicit}
Grigorii~A Margulis.
\newblock Explicit constructions of graphs without short cycles and low density
  codes.
\newblock {\em Combinatorica}, 2(1):71--78, 1982.
\newblock URL: \url{https://doi.org/10.1007/BF02579283}.

\bibitem[ME12]{mari2012positive}
A.~Mari and J.~Eisert.  
\newblock Positive Wigner functions render classical simulation of quantum computation efficient.  
\newblock {\em Physical Review Letters}, 109(23):230503, 2012.


\bibitem[NVY18]{nvy18}
Chinmay Nirkhe, Umesh Vazirani, and Henry Yuen.  
\newblock Approximate low-weight check codes and circuit lower bounds for noisy ground states.  
\newblock In {\em Proceedings of ICALP 2018}, Schloss Dagstuhl – Leibniz-Zentrum für Informatik, 2018.  
\newblock URL: \url{https://drops.dagstuhl.de/entities/document/10.4230/LIPIcs.ICALP.2018.91}.


\bibitem[ODMZ22]{oszmaniec+:FS}
Micha\l{} Oszmaniec, Ninnat Dangniam, Mauro~E.S. Morales, and Zolt\'an
  Zimbor\'as.
\newblock Fermion sampling: A robust quantum computational advantage scheme
  using fermionic linear optics and magic input states.
\newblock {\em PRX Quantum}, 3:020328, May 2022.
\newblock URL: \url{https://link.aps.org/doi/10.1103/PRXQuantum.3.020328}.


\bibitem[SBMW19]{setia2019superfast}
Kanav Setia, Sergey Bravyi, Antonio Mezzacapo, and James~D Whitfield.
\newblock Superfast encodings for fermionic quantum simulation.
\newblock {\em Physical Review Research}, 1(3):033033, 2019.
\newblock URL: \url{https://doi.org/10.1103/PhysRevResearch.1.033033}.


\bibitem[TD02]{Terhal_2002}
Barbara~M. Terhal and David~P. DiVincenzo.
\newblock Classical simulation of noninteracting-fermion quantum circuits.
\newblock {\em Physical Review A}, 65(3), Mar 2002.
\newblock URL: \url{https://doi.org/10.1103%2Fphysreva.65.032325}.


\bibitem[WFC23]{WFC:guide}
Jordi Weggemans, Marten Folkertsma, and Chris Cade.
\newblock Guidable local hamiltonian problems with implications to heuristic
  ans\"atze state preparation and the quantum pcp conjecture.
\newblock {\em arXiv preprint arXiv:2302.11578}, 2023.
\newblock URL: \url{https://doi.org/10.48550/arXiv.2302.11578}.

\end{thebibliography}
\end{document}